\documentclass[conference,letterpaper]{IEEEtran}
\ifCLASSINFOpdf
\else
\fi
\usepackage[cmex10]{amsmath}
\usepackage{amsthm}
\usepackage{amssymb}
\usepackage{bm}
\usepackage[final]{graphicx}
\usepackage{amsfonts}

\usepackage{subfigure}
\usepackage{caption}

\newtheorem{lemma}{Lemma}
\newtheorem{theorem}{Theorem}

\begin{document}
%
\title{Single-Source/Sink Network Error Correction \\Is as Hard as Multiple-Unicast}

\author{\IEEEauthorblockN{Wentao Huang and Tracey Ho}
\IEEEauthorblockA{Department of Electrical Engineering\\
California Institute of Technology\\
Pasadena, CA\\
\{whuang,tho\}@caltech.edu}
\and
\IEEEauthorblockN{Michael Langberg}
\IEEEauthorblockA{Department of Electrical Engineering\\
University at Buffalo, SUNY\\
Buffalo, NY\\
mikel@buffalo.edu}
\and
\IEEEauthorblockN{Joerg Kliewer}
\IEEEauthorblockA{Department of ECE\\
New Jersey Institute of Technology\\
Newark, NJ\\
jkliewer@njit.edu}}


%


\maketitle

\begin{abstract}
We study the problem of communicating over a
single-source single-terminal network in the
presence of an adversary that may jam a single link of the
network. If any one of the edges can be jammed, the capacity of
such networks is well understood and follows directly from
the connection between the minimum cut and maximum flow in
single-source single-terminal networks. In this work we consider
networks in which some edges cannot be jammed, and show that
determining the network communication capacity is at least as hard as solving the
multiple-unicast network coding problem for the
error-free case. The latter problem is a long standing open
problem.
\end{abstract}


%
\IEEEpeerreviewmaketitle

\allowdisplaybreaks[4]
\section{Introduction}
The problem of network error correction concerns reliable transmission of information in a network with point-to-point noiseless channels, in the presence of an adversary. The adversary controls a set $A$ of channels in the network and may corrupt the information transmitted on these channels in an arbitrary way. A network error correction code, first introduced by Cai and Yeung \cite{Yeung:2006ut, Yeung:2006vl}, is a network code that can correct adversarial errors injected into the network from a set of channels $A$, for all $A \in \mathcal{A}$, where $\mathcal{A}$ is a prescribed collection of subsets of channels that characterizes the strength of the adversary. For single-source multicast, under the simplifying assumption that all channels  have unit capacity and $\mathcal{A}$ is the collection of all subsets containing $Z$ channels, \cite{Yeung:2006ut, Yeung:2006vl} show that the cut-set bound is tight and characterizes the network error correction capacity, which can be achieved by a linear code. Under similar settings a variety of works, e.g.,  \cite{Koetter:2008jt, Jaggi:2008dq, Kschischang:2008jj,Zhang:2008wf,Nutman:2008ve,Huang:2013ip} have proposed different efficient capacity-achieving codes and strategies.

 However, under slightly more general settings such that channel capacities are non-uniform or $\mathcal{A}$ has a more general structure, much less is known about the network error correction capacity and achievable strategies. Kim et al. \cite{Kim:2011ec} study a model in which channel capacities are arbitrary and show that capacity upper bounds based on cut-set approaches are generally not tight. \cite{Kim:2011ec} also constructs examples where  linear codes are  insufficient to achieve capacity. Kosut et al. \cite{Kosut:2009ts} study a model in which the adversary controls network nodes instead of channels, which is a special case of network error correction for non-uniform $\mathcal{A}$  in the sense that $\mathcal{A}$ may include subsets of different sizes. In this case  \cite{Kosut:2009ts} constructs an example that linear codes are inadequate to achieve capacity. Achievable strategies under this node adversary model are also studied in \cite{DaWang:2010wv,Kosut:2010ti,Che:2013vy}, whereas determining the capacity region remains an open problem. As opposed to the well studied and well understood setting of \cite{Yeung:2006ut, Yeung:2006vl}, the subtlety of finding and achieving  the network error correction capacity in the more general settings above motivates us to examine the fundamental complexity of the general network error correction problem.

 In this paper, we show that solving the single-unicast network error correction problem with general $\mathcal{A}$ is as hard as solving the multiple-unicast network coding problem (with no error).
Specifically, we convert any unit rate $k$-unicast network coding problem into a corresponding network error correction problem with a single source, a single sink, and a single adversarial channel chosen from a subset of channels, such that the unit rate $k$-unicast is feasible with zero error if and only if the zero-error network error correction capacity is $k$. Under the vanishing error model, we show a similar but slightly weaker result. Specifically, in this case if the unit rate $k$-unicast is feasible, then a network error correction rate of $k$ is feasible. Conversely, if a network error correction rate of $k$ is feasible, then the unit rate $k$-unicast is asymptotically feasible.


Our results add to the portfolios of problems that are connected to multiple-unicast network coding, which is a long standing open problem not presently known to be in P, NP or undecidable \cite{Lehman:2005wi, Medard:2009uj, Dougherty:2005vf,Dougherty:2007um}. Previously, equivalence results have been established, e.g., between multiple-unicast network coding and multiple-multicast network coding \cite{Dougherty:2006vl, Wong:2013ve}, index coding \cite{ElRouayheb:2010tp, Effros:2013ty}, secure network coding \cite{Chan:2008uf, Huang:2013br} and two-unicast network coding \cite{Kamath:2014wy}. 

 The remainder of the paper is structured as follows. In Section II, we present the models and definitions of multiple-unicast network coding and single-source single-sink network error correction. In Section III and IV, we prove the reduction from multiple-unicast to network error correction for the zero error model and vanishing error model, respectively. Finally, we conclude the paper in Section V.

\section{Models}
\subsection{Multiple-unicast Network Coding}
We model the network to be a directed graph $\mathcal{G}=(\mathcal{V},\mathcal{E})$, where the set of vertices $\mathcal{V}$ represents network nodes and the set of edges $\mathcal{E}$ represents network channels. Each edge $e \in \mathcal{E}$ has a capacity $c_e$, which is the maximum number of bits\footnote{For convenience we assume that the network channels transmit binary symbols. Our results can be naturally extended to the general $q$-ary case.} that can be transmitted on $e$ in one transmission. An instance  $\mathcal{I}=(\mathcal{G}, \mathcal{S}, \mathcal{T}, B)$ of the  \emph{multiple-unicast network coding problem},  includes a network $\mathcal{G}$, a set of source nodes $\mathcal{S} \subset \mathcal{V}$, a set of terminal nodes $\mathcal{T} \subset \mathcal{V}$ and an $|\mathcal{S}|$ by $|\mathcal{T}|$ requirement matrix $B$. The $(i,j)$-th entry of $B$ equals 1 if terminal $j$ requires the information from source $i$ and equals 0 otherwise. We assume that $B$ is a permutation matrix and so each source is paired with a single terminal. Let $s(t)$ be the source that is required by terminal $t$. Denote $[n] \triangleq \{1,.., \lceil n \rceil\}$, then each source $s \in \mathcal{S}$ is associated with a  message, which is a rate $R_s$ random variable $M_s$ uniformly distributed over $[2^{nR_s}]$. The messages for different sources are independent. A network code of length $n$ is defined as a set of encoding functions $\phi_e$ for every $e \in \mathcal{E}$ and a set of decoding functions $\phi_t$ for each $t \in \mathcal{T}$. For each $e =(u,v) $, the encoding function $\phi_e$ is a function taking as input the random variables associated with incoming edges of node $u$ and the random variable $M_u$ if $u \in \mathcal{S}$, and maps to values in $[2^{n c_e}]$. For each $t \in \mathcal{T}$, the decoding function $\phi_t$ maps all random variables associated with the incoming edges of $t$, to a message $\hat{M}_{s(t)}$ with values in  $[2^{nR_{s(t)}}]$.

A network code $\{\phi_e, \phi_t \}_{e \in \mathcal{E}, t \in \mathcal{T}}$ is said to \emph{satisfy} a terminal $t$ under transmission $(m_s, s\in\mathcal{S})$ if $\hat{M}_{s(t)} = m_{s(t)}$ when $(M_s, s \in \mathcal{S}) = (m_s, s \in \mathcal{S})$. A network code is said to satisfy the multiple-unicast network coding problem $\mathcal{I}$ with error probability $\epsilon$ if the probability that all $t \in \mathcal{T}$ are simultaneously  satisfied is at least $1-\epsilon$. The probability is taken over the joint distribution on random variables $(M_s, s\in\mathcal{S})$. Namely, the network code satisfies $\mathcal{I}$ with error probability $\epsilon$ if
\begin{align}
\Pr_{(M_s, s \in \mathcal{S})} \left\{ \bigcap_{t \in \mathcal{T}} t \text{ is satisfied under }(M_s, s \in \mathcal{S})  \right\} \ge 1-\epsilon
\end{align}

For an instance $\mathcal{I}$ of the multiple-unicast network coding problem, rate $R$ is said to be \emph{feasible} if $R_s = R$, $\forall s\in\mathcal{S}$, and for any $\epsilon>0$, there exists a network code with sufficiently large length that satisfies $\mathcal{I}$ with error probability at most $\epsilon$. Rate $R$ is said to be \emph{feasible with zero error} if $R_s = R$, $\forall s\in\mathcal{S}$ and there exists a network code that satisfies $\mathcal{I}$ with zero error probability. Rate $R$ is said to be \emph{asymptotically feasible} if  for any $\delta > 0$, rate $(1-\delta)R$ is feasible. The capacity of $\mathcal{I}$ refers to  the supremum over all rates $R$ that are asymptotically feasible and the zero-error capacity of $\mathcal{I}$ refers to the  supremum over all rates $R$ that are feasible with zero error. The given model assumes all sources transmit information at equal rate. There is no loss of generality in this assumption as a varying rate source $s$ can be modeled by several equal rate sources co-located at $s$.

\subsection{Single-Source Single-Sink Network Error Correction}
An instance $\mathcal{I}_c = (\mathcal{G}, s, t, \mathcal{A})$ of the \emph{single-source single-terminal network error correction problem} includes a network $\mathcal{G}$, a source node $s \in \mathcal{V}$, a terminal node $t \in \mathcal{V}$ and a collection of subsets of channels $\mathcal{A} \subset 2^\mathcal{E}$ susceptible to errors. In this problem the channels are not always reliable and an error is said to occur in a channel if the output of the channel is different from the input. More precisely, the output of a channel $e$ is the input signal superposed by an error signal $r_e$, and we say there is an error on channel $e$ if $r_e \ne 0$. For a subset $A \in \mathcal{A}$ of channels,  an $A$-error is said to occur if an error occurs in every channel in $A$. For an instance $\mathcal{I}_c$ of the single-source single-terminal network error correction problem, a network code $\{\phi_e, \phi_t, \}_{e \in \mathcal{E}, t \in \mathcal{T} }$ is said to \emph{satisfy} a terminal $t$ under transmission $m_s$ if $\hat{M}_{s} = m_{s}$ when $M_s = m_s$, given the occurrence of any error pattern $\bm{r}=(r_e, e \in \mathcal{E} )$ that results in an $A$-error, for all $A \in \mathcal{A}$. A network code is said to satisfy problem $\mathcal{I}_c$ with error probability $\epsilon$ if the probability that $t$ is  satisfied is at least $1-\epsilon$. The probability is taken over the  distribution on the source message $M_s$.

For an instance $\mathcal{I}_c$ of the single-source single-terminal network error correction problem, rate $R$ is said to be \emph{feasible} if $R_s = R$ and for any $\epsilon>0$, there exists a network code with sufficiently large length that satisfies $\mathcal{I}_c$ with error probability at most $\epsilon$. Rate $R$ is said to be \emph{feasible with zero error} if $R_s = R$ and there exists a network code that satisfies $\mathcal{I}_c$ with zero error probability. Rate $R$ is said to be \emph{asymptotically feasible} if  for any $\delta > 0$, rate $(1-\delta)R$ is feasible. The capacity of $\mathcal{I}_c$ refers to the supremum over all rates $R$ that are asymptotically feasible and the zero-error capacity of $\mathcal{I}_c$ refers to the  supremum over all rates $R$ that are feasible with zero error.

Throughout the paper we denote by $\mathcal{R}_{\mathcal{A}}$ the set of all possible error patterns $\bm{r}$ that result in $A$-errors, where $A \in \mathcal{A}$.


\section{Reduction from Multiple-unicast to Network Error Correction: Zero Error Case}
In this section we reduce the multiple-unicast network coding problem (with no error) to the  single-source single-terminal network error correction problem with at most a single adversarial channel. We start with the zero-error case.
\begin{figure}[h!]
  \begin{center}
      \includegraphics[width=0.30\textwidth]{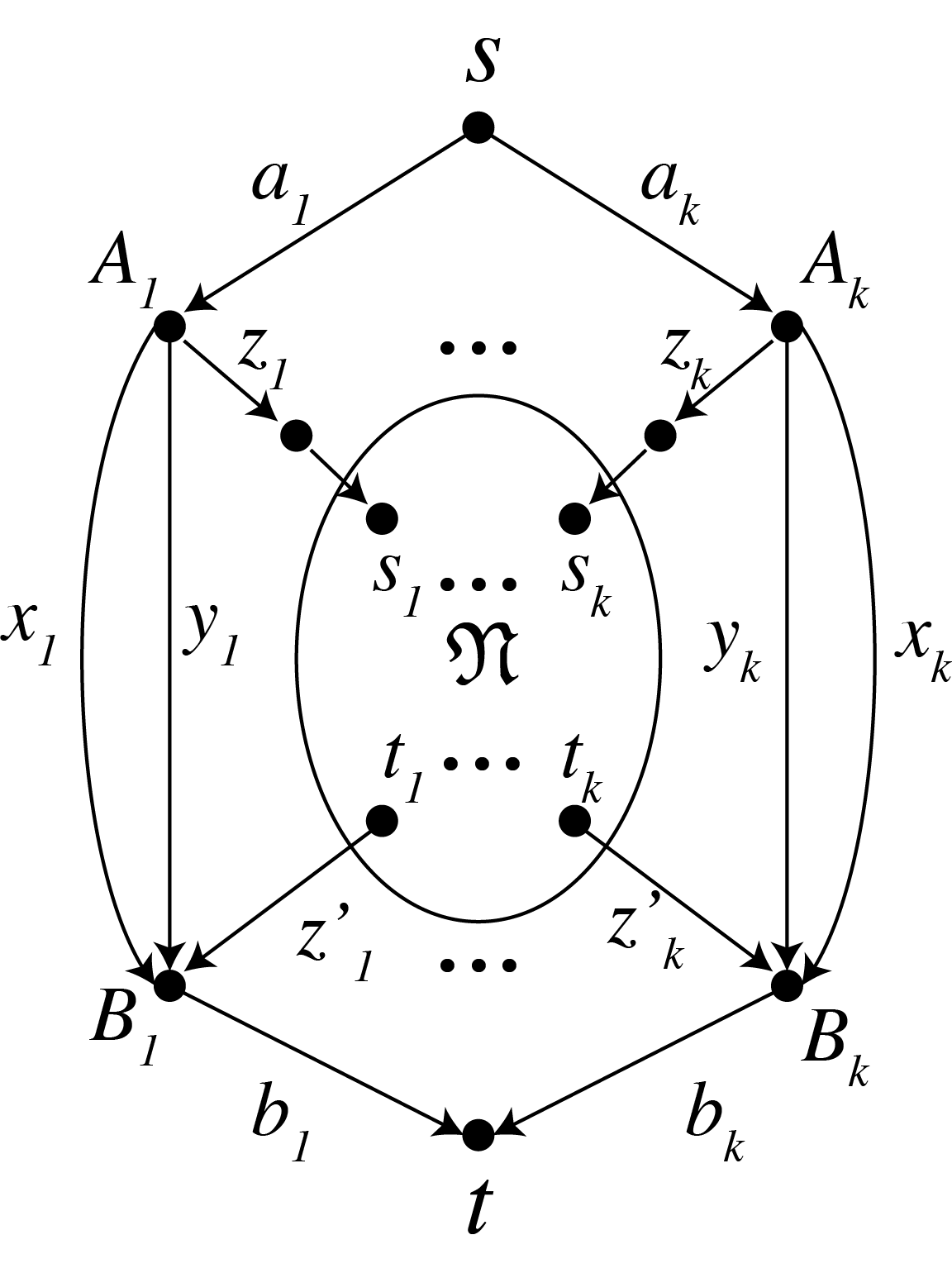}
  \caption{In the single-source single-terminal network error correction problem $\mathcal{I}_c$, the source $s$ wants to communicate with the terminal $t$. $\mathcal{N}$ is a general network with point-to-point noiseless channels. All edges outside $\mathcal{N}$ (i.e., edges for which at least one of its end-point does not belong to $\mathcal{N}$) have unit capacity. There is at most one error in this network, and this error can occur at any edge except $\{ a_i, b_i, 1 \le i \le k \}$. Namely, $\mathcal{A}$ includes all singleton sets of a single edge in the network except $\{a_i\}$ and $\{b_i\}$, $i=1,...,k$. Note that there are $k$ branches in total but only the first and the $k$-th branches are drawn explicitly. The multiple-unicast network coding problem $\mathcal{I}$ is defined on the network $\mathcal{N}$, where the $k$ source-destination pairs are $(s_i,t_i), i=1,...,k$, and all channels are error-free.}\label{zeroerr}
           \end{center}
\end{figure}
\begin{theorem}
Given any multiple-unicast network coding problem $\mathcal{I}$ with source-destination pairs $\{ (s_i, t_i), i=1,...,k \}$, a corresponding single-source single-sink network error correction problem $\mathcal{I}_c=(\mathcal{G},s,t,\mathcal{A})$ in which $\mathcal{A}$ includes sets with at most one edge can be constructed as specified in Figure \ref{zeroerr},  such that  the zero-error  capacity of $\mathcal{I}_c$ is $k$ if and only if unit rate is feasible with zero error in $\mathcal{I}$.
\end{theorem}
\begin{proof}
The zero-error capacity of $\mathcal{I}_c$ is upper bounded by $k$, because it is the min-cut from $s$ to $t$.

``$\Rightarrow$'' We show that the feasibility of a zero-error rate $k$ in $\mathcal{I}_c$ implies the feasibility of unit zero-error rate in $\mathcal{I}$.

Suppose a zero-error rate $k$ is achieved in $\mathcal{I}_c$ by a network code with length $n$, and denote the source message by $M$, then $M$ is uniformly distributed over $[2^{nk}]$. For any edge $e \in \mathcal{E}$, we denote by $e(m,\bm{r}) : [2^{nk}] \times \mathcal{R}_{\mathcal{A}} \to [2^n]$  the signal received on edge $e$ when the source message equals $m$ and the error pattern $\bm{r}$ occurs in the network. When the context is clear, we may denote $e(m,\bm{r})$ simply by $e$.

Let $\bm{b}(m,\bm{r}) = (b_1(m,\bm{r}), ..., b_k(m,\bm{r}))$, then because the edges $b_1,...,b_k$ form a cut-set from $s$ to $t$, $\bm{b}(m,\bm{r})$ must be injective with respect to $m$  due to the zero error decodability constraint. Formally, for two different messages $m_1 \ne m_2$, it follows from the zero error decodability constraint that $\bm{b}(m_1,\bm{r}_1) \ne \bm{b}(m_2,\bm{r}_2)$, $\forall \bm{r}_1, \bm{r}_2 \in \mathcal{R}_{\mathcal{A}}$. Note that the codomain of $\bm{b}$ is $[2^n]^k$, which has the same size as the set of messages $[2^{nk}]$. Therefore denote by $\bm{b}(m) \triangleq  \bm{b}(m,\bm{0})$, then $\bm{b}(m)$ is a bijective function and $\bm{b}(m,\bm{r}) = \bm{b}(m)$, $\forall \bm{r} \in \mathcal{R}_{\mathcal{A}}$.
Similarly $ \bm{a}(m,\bm{r}) = (a_1(m,\bm{r}), ..., a_k(m,\bm{r}))$ is also a bijective function of the message, regardless of the error patterns.



For any $e \in \mathcal{E}$, denote $e(m) \triangleq e(m,\bm{0})$. For $i=1,...,k$,
we claim that for any two messages $m_1, m_2 \in [2^{nk}]$ such that $a_i(m_1) \ne a_i(m_2)$, it follows that $x_i(m_1) \ne x_i(m_2)$, $y_i(m_1) \ne y_i(m_2)$ and $z_i(m_1) \ne z_i(m_2)$. Suppose for contradiction that there exist $m_1, m_2$ such that $a_i(m_1) \ne a_i (m_2)$ and such that the claim is not true, i.e.,
$x_i(m_1) = x_i(m_2)$ or $y_i(m_1) = y_i(m_2)$ or $z_i(m_1) = z_i(m_2)$. First consider the case that $x_i(m_1) = x_i(m_2)$.
Because of the one-to-one correspondence between $m$ and $\bm{a}$, there exists a message $m_3 \ne m_1, m_2$ and such that $\bm{a}(m_3) = (a_1(m_1), ..., a_{i-1}(m_1),a_i(m_2),a_{i+1}(m_1), ..., a_k(m_1))$. Then $x_i(m_1)=x_i(m_3)$ because by hypothesis $x_i(m_1) = x_i(m_2)$.
Consider the following two scenarios. In the first scenario, $m_1$ is transmitted, and an error turns $y_i(m_1)$ into $y_i(m_3)$; in the second scenario,  $m_3$ is transmitted, and an error turns $z_i(m_3)$ into $z_i(m_1)$. Then the cut-set signals $a_1,...,a_{i-1}, x_{i}, y_i, z_i, a_{i+1},...,a_k$ are exactly the same in both scenarios, and so it is impossible for $t$ to distinguish $m_1$ from $m_3$, a contradiction to the zero error decodability constraint.
Therefore $x_i(m_1) \ne x_i(m_2)$. With a similar argument it follows that $y_i(m_1) \ne y_i(m_2)$ and $z_i(m_1) \ne z_i(m_2)$, and the claim is proved.

The claim above suggests that  $x_i, y_i$ and $z_i$, as functions of $a_i$, are injective. They are also surjective functions because the domain and codomain are both $[2^n]$.
Hence there are one-to-one correspondences between $a_i$, $x_i$, $y_i$ and $z_i$.

 Next we show that for any two messages $m_1, m_2$, if $b_i(m_1) \ne b_i(m_2)$, then $z_i'(m_1) \ne z_i'(m_2)$. Suppose for contradiction that there exists $m_1 \ne m_2$ such that $b_i(m_1) \ne b_i(m_2)$ and $z_i'(m_1) = z_i'(m_2)$.
Then if $m_1$ is transmitted and an error $\bm{r}_1$ turns $x_i(m_1)$ into $x_i(m_2)$, the node $B_i$ will receive the same signals as in the case that $m_2$ is transmitted and an error $\bm{r_2}$ turns $y_i(m_2)$ into $y_i(m_1)$. Therefore $b_i(m_1,\bm{r}_1) = b_i(m_2, \bm{r}_2)$. But, as shown above, because $b_i(m_1, \bm{r}_1) = b_i(m_1)$ and $b_i(m_2, \bm{r}_2) = b_i(m_2)$, it follows that $b_i(m_1) = b_i(m_2)$, a contradiction.
This suggests that if $z'_i(m_1) = z'_i(m_2)$ then $b_i(m_1) = b_i(m_2)$ and therefore $b_i$ is a function of $z_i'$. The function is surjective because $b_i$ takes all $2^n$ possible values. Then since the domain and the codomain are both $[2^n]$, it follows that $b_i$ must be a bijective function of $z_i'$. With the same argument it follows that $b_i$ is also a bijective function of $x_i$.

Hence $z_i$ is a bijection of $a_i$, $a_i$ is a bijection of $x_i$, $x_i$ is a bijection of $b_i$, and $b_i$ is a bijection of $z_i'$. Therefore for all $1 \le i \le k$, $z_i$ is a bijection of $z_i'$, and therefore unit rate is feasible with zero error in $\mathcal{I}$.

``$\Leftarrow$'' Conversely, we show that the feasibility of the unit zero-error rate in $\mathcal{I}$ implies the achievability of a zero-error rate of $k$ in $\mathcal{I}_c$. A constructive scheme is shown in Figure \ref{ach}. In $\mathcal{I}_c$, the source lets $M = (M_1,...,M_k)$, where the $M_i$'s are i.i.d. uniformly distributed over $[2^n]$. Let the network code be $a_i(M)=x_i(M)=y_i(M)=z_i(M)=z'_i(M)=M_i$, $i=1,...,k$, and let node $B_i$, $i=1,...,k$, perform majority decoding. It is straightforward to see that the scheme ensures that $b_i(M) = M_i$ under all error patterns in $\mathcal{R}_{\mathcal{A}}$. Therefore rate $k$ is feasible with zero error in $\mathcal{I}_c$. This rate achieves capacity since it is equal to the min-cut from $s$ to $t$.


\end{proof}

\begin{figure}[h!]
  \begin{center}
      \includegraphics[width=0.30\textwidth]{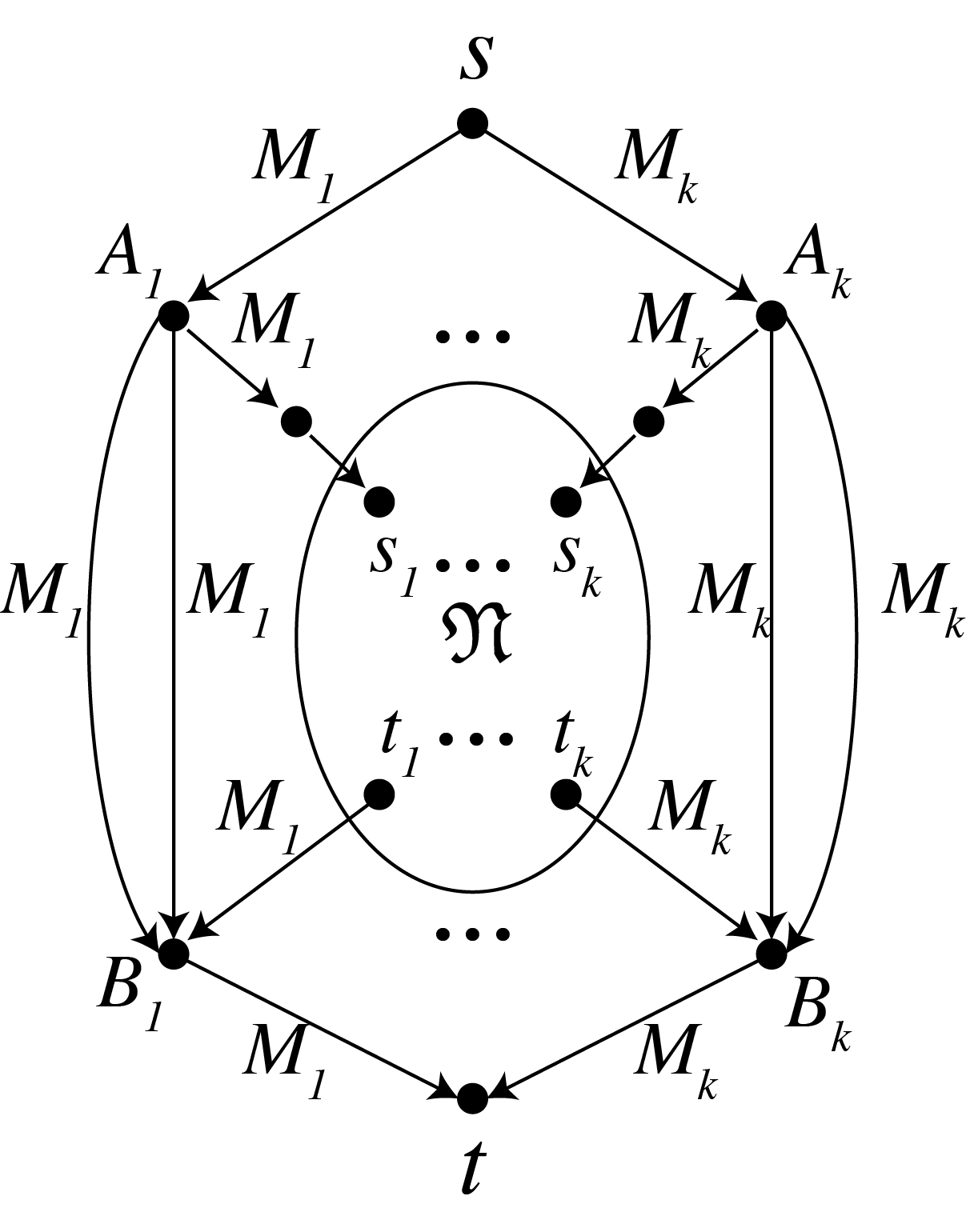}
  \caption{A scheme to achieve zero-error rate $k$ in $\mathcal{I}_c$ given that unit rate is feasible with zero error in $\mathcal{I}$. $M=(M_1,...,M_k)$ and node $B_i$ performs majority decoding.}\label{ach}
           \end{center}
\end{figure}

\section{Reduction from Multiple-unicast to Network Error Correction: Vanishing Error Case}
In this section we show that a similar but slightly weaker result holds under the vanishing error model.
\begin{theorem}\label{epsilonerrk2}
Given any multiple-unicast network coding problem $\mathcal{I}$ with source-destination pairs $\{ (s_i, t_i), i=1,...,k \}$, a corresponding single-source single-sink network error correction problem $\mathcal{I}_c=(\mathcal{G},s,t,\mathcal{A})$  in which $\mathcal{A}$ includes sets with at most a single edge can be constructed as specified in Figure \ref{zeroerr},  such that  if unit rate is feasible in $\mathcal{I}$ then rate $k$ is feasible in $\mathcal{I}_c$. Conversely, if rate $k$ is feasible in $\mathcal{I}_c$ then unit rate is asymptotically feasible in $\mathcal{I}$.
\end{theorem}
We first give a simple lemma.
\begin{lemma}\label{info}
Let $X,Y,Z$ be three arbitrary random variables. Then
\begin{align*}
I(X;Z) \ge I(X;Y) + I(Y;Z) - H(Y).
\end{align*}
\end{lemma}
\begin{proof}
\begin{align*}
I(X;Z) & = H(Z) - H(Z|X)\\
& \ge H(Z) - H(Z,Y|X)\\
& = H(Z) - H(Y|X) - H(Z|Y,X)\\
& \ge H(Z) - H(Y|X) -H(Z|Y)\\
& = H(Y) - H(Y|X) + H(Z) - H(Z|Y) -H(Y)\\
& = I(X;Y) + I(Y;Z) - H(Y)
\end{align*}
\end{proof}
In the following we prove Theorem \ref{epsilonerrk2}.


``$\Rightarrow$'' We first show the feasibility of rate $k$ in $\mathcal{I}_c$ implies the asymptotic feasibility of unit rate in $\mathcal{I}$.

We will take the following path. In $\mathcal{I}_c$, we apply the network code that achieves rate $k$ with message uniformly distributed over a selected subset of $[2^{nk}]$. Then we show that for $i=1,...,k$, this induces a large mutual information between random variables $a_i$ and $b_i$, between $b_i$ and $z'_i$, and between $z_i$ and $a_i$. Hence it implies a large mutual information between $z_i$ and $z'_i$ and finally we show that this implies the asymptotic feasibility of unit rate in $\mathcal{I}$.\vspace{2mm}

\noindent\textbf{Step 1: Select a subset of messages.}

Suppose in $\mathcal{I}_c$ a rate of $k$ is achieved by a network code with length $n$ and with a probability of error $\epsilon$. Let $M$ be the source message uniformly distributed over $\mathcal{M} = [2^{kn}]$ and let $\hat{M}$ be the output of the decoder at the terminal. Partition $\mathcal{M}$ into good and bad messages $\mathcal{M}^{g} + \mathcal{M}^{b}$ in the way that  $m \in \mathcal{M}^{g}$ if the network code satisfies $t$ under transmission $m$, i.e., the terminal decodes successfully $\hat{M} = m$ when $M=m$ for all $\bm{r} \in \mathcal{R}_{\mathcal{A}}$.
Therefore if $m \in \mathcal{M}^{b}$ then there exists $\bm{r} \in \mathcal{R}_{\mathcal{A}}$ such that $\hat{M} \ne m$ if $M=m$ and $\bm{r}$ occurs, i.e., $\bm{r}$ results in a decoding error.
By the hypothesis on the probability of error it follows that $|\mathcal{M}^{b}| \le 2^{kn} \epsilon$.

For $i=1,...,k$, let $x_i(m,\bm{r}): \mathcal{M} \times \mathcal{R}_{\mathcal{A}} \to [2^n]$ be the  signal received from channel $x_i$ when $m$ is transmitted by the source and the error pattern $\bm{r}$ happens. Let $x_i(m) = x_i(m,\bm{0})$, $\bm{x}(m,\bm{r}) = (x_1(m,\bm{r}) , ...,  x_k(m,\bm{r}) )$ and $\bm{x}(m) = (x_1(m), ...,  x_k(m))$. We define functions $a_i, b_i, y_i, z_i, z'_i, \bm{a}, \bm{b}, \bm{y}, \bm{z}, \bm{z}'$ in a similar way.

\begin{lemma}\label{ltwo}
There exists $\mathcal{M}^\circ \subset \mathcal{M}^g$ such that for any $m_1, m_2 \in \mathcal{M}^\circ$, $m_1 \ne m_2$, it follows that $\bm{a}(m_1) \ne \bm{a}(m_2)$, $\bm{b}(m_1) \ne \bm{b}(m_2)$ and $\bm{z}'(m_1) \ne \bm{z}'(m_2)$, and such that $|\mathcal{M}^\circ| \ge 2^{kn}(1-\epsilon')$, where $\epsilon'=4\epsilon$.
\end{lemma}
\begin{proof}
As $\bm{a}$ and $\bm{b}$ are cut-set signals, for any $m_1, m_2 \in \mathcal{M}^g$, $m_1 \ne m_2$, it follows that $\bm{a}(m_1) \ne \bm{a}(m_2)$ and $\bm{b}(m_1) \ne \bm{b}(m_2)$. Setting $\mathcal{B}^g = \{ \bm{b}(m) : m \in \mathcal{M}^g \}$, it holds that $|\mathcal{B}^g| \ge (1-\epsilon) 2^{kn} $, and setting $\mathcal{B}^b = [2^n]^k \backslash \mathcal{B}^g$, it holds that $|\mathcal{B}^b| \le  2^{kn} \epsilon$.

For any $m \in \mathcal{M}^g$, it is said to be a \emph{poor} message if there exists another $m' \in \mathcal{M}^g$ such that $\bm{z}'(m) = \bm{z}'(m')$. Consider an arbitrary poor message $m_1$. By definition there exists $m_2 \in \mathcal{M}^g$, $m_2 \ne m_1$, such that $\bm{z}'(m_1) = \bm{z}'(m_2)$. Since $\bm{b}(m_1) \ne \bm{b}(m_2)$, there exists $j$ such that $b_j(m_1) \ne b_j(m_2)$. Let $\bm{r}_1$ be the error pattern that changes the signal on $x_j$ to $x_j(m_2)$, and let $\bm{r}_2$ be the error pattern that changes the signal on $y_j$ to $y_j(m_1)$. Then if $m_1$ is sent and $\bm{r}_1$ happens, node $B_j$ will receive the same inputs as in the situation that $m_2$ is sent and $\bm{r}_2$ happens. Therefore $b_j(m_1,\bm{r}_1) = b_j(m_2, \bm{r}_2)$ and it follows that either $b_j(m_1, \bm{r}_1) \ne b_j(m_1)$ or $b_j(m_2, \bm{r}_2) \ne b_j(m_2)$. In the former case, the tuple of signals $(b_1(m_1, \bm{r}_1),...,b_j(m_1,\bm{r}_1),...,b_k(m_1, \bm{r}_1) ) = (b_1(m_1),...,b_j(m_1,\bm{r}_1),...,b_k(m_1) )$ will be decoded by the terminal to message $m_1$ (because by hypothesis $m_1 \in \mathcal{M}^g$, which is decodable under any error $\bm{r} \in \mathcal{R}_{\mathcal{A}}$). It is therefore an element of $\mathcal{B}^b$ as it does not equal $\bm{b}(m_1)$.
Similarly, in the latter case,  $(b_1(m_2, \bm{r}_2),...,b_j(m_2,\bm{r}_2),...,b_k(m_2, \bm{r}_2) ) = (b_1(m_2),...,b_j(m_2,\bm{r}_2),...,b_k(m_2) )$ will be decoded by the terminal to message $m_2$ and is an element of $\mathcal{B}^b$. For $\hat{\bm{z}} \in [2^n]^k$, let $\mathcal{M}({\hat{\bm{z}}'})$ be the set of messages $\{m \in \mathcal{M}^g :  \bm{z}'(m) = \hat{\bm{z}}' \}$. Then if $|\mathcal{M}({\hat{\bm{z}}'})| > 1$, by the argument above, there are at least $\left\lfloor \frac{|\mathcal{M}({\hat{\bm{z}}'})|}{2} \right \rfloor \ge \frac{|\mathcal{M}({\hat{\bm{z}}'})|}{3} $ elements of $\mathcal{B}^b$, such that each of them will be decoded by the terminal to some message $m \in \mathcal{M}(\hat{\bm{z}}')$. Let $\mathcal{M}^{\text{poor}}$ be the set of all poor messages, then
\begin{align*}
|\mathcal{M}^{\text{poor}}| = \sum_{\hat{\bm{z}}': |\mathcal{M} (\hat{\bm{z}}') |>1 } |\mathcal{M} (\hat{\bm{z}}') | \le 3|\mathcal{B}^b| \le 3\epsilon \cdot 2^{kn} .
\end{align*}
Let $\mathcal{M}^\circ = \mathcal{M}^g \backslash \mathcal{M}^{\text{poor}}$, then $|\mathcal{M}^\circ| = |\mathcal{M}^g| - |\mathcal{M}^{\text{poor}}| \ge (1-\epsilon')2^{kn}$, where $\epsilon' = 4\epsilon$.  This proves the assertion.
\end{proof}

Let $\mathcal{A}^{\circ} = \{\bm{a}(m) | m \in \mathcal{M}^{\circ}  \}$ and $\mathcal{A}^{\times} = [2^n]^k \backslash \mathcal{A}^{\circ}$, then $|\mathcal{A}^{\circ}| = |\mathcal{M}^{\circ}| \ge (1-\epsilon')2^{kn} $ since by Lemma \ref{ltwo}, $\bm{a}(m_1) \ne \bm{a}(m_2)$ for $m_1, m_2 \in \mathcal{M}^\circ, m_1 \ne m_2$. Therefore $|\mathcal{A}^\times| \le 2^{kn}\epsilon'$. Similarly let $\mathcal{B}^{\circ} = \{\bm{b}(m) | m \in \mathcal{M}^{\circ}  \}$ and $\mathcal{B}^{\times} = [2^n]^k \backslash \mathcal{B}^{\circ}$, then $|\mathcal{B}^{\times}| \le 2^{kn} \epsilon'$.
For $i=1,...,k$ , let $\mathcal{A}^\circ_i = \{ a_i(m) | m \in \mathcal{M}^\circ \}$, then $|\mathcal{A}^\circ| \ge (1-\epsilon')2^{kn}$ implies that $|\mathcal{A}^\circ_i| \ge (1-\epsilon') 2^n $. For $\hat{a}_i \in \mathcal{A}^\circ_i$, let $\mathcal{M}(\hat{a}_i) =\{ m \in \mathcal{M}^\circ : a_i(m) =\hat{a}_i  \}$ and
define $N(\hat{a}_i) = |\mathcal{M}(\hat{a}_i)|$. Furthermore define $\mathcal{A}^{\circ}_{i,l} = \{ \hat{a}_i \in \mathcal{A}^\circ_i | N(\hat{a}_i) \ge (1 - l\epsilon') 2^{(k-1)n}  \}$.

We show that the size of $\mathcal{A}^{\circ}_{i,l}$ is large. Consider any $\hat{a}_i \in \mathcal{A}^\circ_i \backslash \mathcal{A}_{i,l}^{\circ}$, then by definition $|\{ (a_1,...,a_k) \in \mathcal{A}^\circ: a_i = \hat{a}_i  \}| < (1-l\epsilon')2^{(k-1)n}$. And because $|\{(a_1,...,a_k) \in [2^n]^k : a_i = \hat{a}_i \}| = 2^{(k-1)n}$, there are at least $l\epsilon' \cdot 2^{(k-1)n}$ elements of $\mathcal{A}^\times$ such that their $i$-th entry equals to $\hat{a}_i$. Therefore $ |\mathcal{A}^\circ_i \backslash \mathcal{A}_{i,l}^{\circ}|
\cdot l \epsilon'  \cdot 2^{(k-1)n} \le |\mathcal{A}^\times| \le 2^{kn} \epsilon'$, and $|\mathcal{A}^\circ_i \backslash \mathcal{A}_{i,l}^{\circ}| \le 2^n/l$. So $|\mathcal{A}^{\circ}_{i,l} | \ge |\mathcal{A}^\circ_i| -  2^n/l \ge (1 - \epsilon' - 1/l )2^n$. Define  $\mathcal{B}^\circ_i$, $\mathcal{B}^{\circ}_{i,l}$, $\mathcal{Z}'^\circ_i$ and $\mathcal{Z}'^\circ_{i,l}$ similarly, then it follows from the same argument that $  |\mathcal{B}^{\circ}_{i,l}| , |\mathcal{Z}'^\circ_{i,l}|\ge (1 - \epsilon' - 1/l )2^n$. \vspace{2mm}

\noindent\textbf{Step 2: Connect $a_i$ and $b_i$.}

Let $M^\circ$ be the random variable that is uniformly distributed over $\mathcal{M}^\circ$. In the following we show that if $M^\circ$ is the source message then $I(a_i;b_i)/n  \to 1$ as $\epsilon \to 0$, $i=1,...,k$.
We start by lower bounding the entropy $H(b_i)$. Consider any $\hat{b}_i \in \mathcal{B}_{i,l}^{ \circ}$, then
\begin{align}
\Pr \{ b_i = \hat{b}_i \} = \frac{N(\hat{b}_i)}{|\mathcal{M}^\circ|} & \ge \frac{(1-l\epsilon')2^{(k-1)n}}{2^{kn}} = \frac{1-l\epsilon'}{2^n} \label{nbl}\\
\Pr \{ b_i = \hat{b}_i \} = \frac{N(\hat{b}_i)}{|\mathcal{M}^\circ|}  & \le  \frac{2^{(k-1)n}}{(1 - \epsilon')2^{kn}} = \frac{1}{2^n(1-\epsilon')} \label{nbu}
\end{align}
Therefore,
\begin{align}
  H(b_i) & = - \sum_{\hat{b}_i \in \mathcal{B}_i^\circ} \Pr \{ \hat{b}_i \} \log (\Pr \{ \hat{b}_i \}) \nonumber \\
  & \ge - \sum_{\hat{b}_i \in \mathcal{B}^{\circ}_{i,l}} \Pr \{ \hat{b}_i \} \log (\Pr \{ \hat{b}_i \}) \nonumber \\
  &  \stackrel{(a)}{\ge} - |\mathcal{B}_{i,l}^{\circ}| \cdot \frac{1-l\epsilon'}{2^n} \cdot \log\left(\frac{1}{2^n(1-\epsilon')}\right) \nonumber \\
  & \ge  \left(1 - \epsilon' - \frac{1}{l} \right)2^n \cdot \frac{1-l\epsilon'}{2^n}\cdot \log(2^n(1-\epsilon')) \nonumber \\
  & = \left(1 - \epsilon' - \frac{1}{l} \right) (1-l\epsilon')(n + \log(1-\epsilon')) \label{hb}
\end{align}
where (a) is due to (\ref{nbl}) and (\ref{nbu}). Similarly it follows that,
\begin{align}
H(a_i) & \ge \left(1 - \epsilon' - \frac{1}{l} \right) (1-l\epsilon')(n + \log(1-\epsilon')) \label{ha}\\
H(z'_i) & \ge \left(1 - \epsilon' - \frac{1}{l} \right) (1-l\epsilon')(n + \log(1-\epsilon')), \label{hz'}
\end{align}

In the next step we upper bound $H(b_i|a_i)$. Recall that $\mathcal{M}(\hat{a}_i) = \{ m \in \mathcal{M}^\circ: a_i(m) = \hat{a}_i \}$, we first prove a useful lemma.
\begin{lemma}\label{three}
Suppose $\{b_i(m):m \in \mathcal{M}(\hat{a}_i)\} =\{\hat{b}_i^{(1)},...,\hat{b}_i^{(L)}\}$, then there exist $(L-1)N(\hat{a}_i)$ distinct elements of $\mathcal{B}^\times$ such that each of them will be decoded by the terminal to some message $m \in \mathcal{M}(\hat{a}_i)$.
\end{lemma}
\begin{proof}
Consider arbitrary $\hat{a}_i \in \mathcal{A}^\circ_i$, by hypothesis there exist $L$ messages $m_1, ..., m_L \in \mathcal{M}(\hat{a}_i)$ such that $b_i(m_j) = \hat{b}_i^{(j)}$, $j=1,...,L$. For $j =1,...,L$, let $\bm{r}_j$ be the error pattern that changes the signal on $z'_i$ to be $z'_i(m_j)$. Then if an arbitrary message $m_0 \in \mathcal{M}(\hat{a}_i)$ is transmitted by the source and $\bm{r}_j$ happens, the node $B_i$ will receive the same inputs as in the situation that $m_j$ is sent and no error happens. Therefore $b_i(m_0, \bm{r}_j) = \hat{b}_i^{(j)}$, and so $\bm{b}(m_0, \bm{r}_j)$ takes $L$ distinct values for $j=1,..,L$. Since $m_0 \in \mathcal{M}^\circ$ is decodable under any error pattern $\bm{r} \in \mathcal{R}_{\mathcal{A}}$, it follows that $\bm{b}(m_0, \bm{r}_j)$ will be decoded by the terminal to $m_0$ for all $j=1,...,L$. Among these $L$ values, i.e., $\{ \bm{b}(m_0, \bm{r}_j) , j=1,...,L \}$, only one is equal to $\bm{b}(m_0)$, and the remaining $L-1$ of them are elements of $\mathcal{B}^\times$. Sum over all $ m_0 \in \mathcal{M}(\hat{a}_i)$ and the assertion is proved.
\end{proof}

Partition $\mathcal{A}^\circ_{i, l}$ into $\mathcal{A}^\circ_{i, L=1} + \mathcal{A}^\circ_{i, L>1}$, such that every element of $\mathcal{A}^\circ_{i, L=1}$ has a corresponding $L=1$ as defined in Lemma \ref{three}. Then it follows from Lemma \ref{three} that:
\begin{align*}
|\mathcal{A}^\circ_{i,L>1}| \cdot (1-l\epsilon') \cdot 2^{(k-1)n} \le |\mathcal{B}^\times| \le 2^{kn} \epsilon'
\end{align*}
and so
\begin{align}\label{L>1}
 |\mathcal{A}^\circ_{i,L>1}| \le \frac{\epsilon' \cdot 2^n}{1-l\epsilon'}.
\end{align}



We are ready to upper bound $H(b_i|a_i)$.
\begin{align}
H(b_i|a_i) & = - \sum_{\hat{a}_i \in \mathcal{A}^\circ_i} \Pr\{ \hat{a}_i \} \sum_{\hat{b}_i \in \mathcal{B}^\circ_{i}} \Pr \{ \hat{b}_i | \hat{a}_i \} \log \Pr\{ \hat{b}_i | \hat{a}_i \}\nonumber\\
& \le I_1 + I_2 + I_3 \label{i1i2i3},
\end{align}
where
\begin{align*}
I_1 &= - \sum_{\hat{a}_i \in \mathcal{A}^\circ_i \backslash \mathcal{A}_{i,l}^\circ} \Pr\{ \hat{a}_i \} \sum_{\hat{b}_i \in \mathcal{B}^\circ_{i}} \Pr \{ \hat{b}_i | \hat{a}_i \} \log \Pr\{ \hat{b}_i | \hat{a}_i \}\\
I_2 &= - \sum_{\hat{a}_i \in \mathcal{A}^\circ_{i,L>1} } \Pr\{ \hat{a}_i \} \sum_{\hat{b}_i \in \mathcal{B}^\circ_{i}  } \Pr \{ \hat{b}_i | \hat{a}_i \} \log \Pr\{ \hat{b}_i | \hat{a}_i \}\\
I_3 &= - \sum_{\hat{a}_i \in \mathcal{A}_{i,L=1}^\circ} \Pr\{ \hat{a}_i \} \sum_{\hat{b}_i \in \mathcal{B}_{i}^\circ} \Pr \{ \hat{b}_i | \hat{a}_i \} \log \Pr\{ \hat{b}_i | \hat{a}_i \}
\end{align*}
We now bound $I_1$, $I_2$ and $I_3$ respectively.
\begin{align}
I_1 &  \le  \sum_{\hat{a}_i \in \mathcal{A}_i^\circ \backslash \mathcal{A}_{i,l}^\circ} \Pr\{ \hat{a}_i \} \log |2^n| \nonumber\\
& = n  \sum_{\hat{a}_i \in \mathcal{A}^\circ_i \backslash \mathcal{A}_{i,l}^\circ} \Pr\{ \hat{a}_i \}\nonumber\\
& \le n  \frac{|\mathcal{M}^\circ | - \sum_{\hat{a}_i \in \mathcal{A}_{i,l}^\circ} N(\hat{a}_i) }{|\mathcal{M}^\circ|} \nonumber\\
& \le n\left(1 - \frac{|\mathcal{A}^\circ_{i,l}| \cdot (1-l\epsilon')2^{(k-1)n} }{|\mathcal{M}^\circ|}\right) \nonumber \\
& \le n  \left(1 - \frac{(1-1/l - \epsilon')2^n (1-l\epsilon')2^{(k-1)n}}{2^{kn}} \right) \nonumber\\
& < (1/l +l\epsilon' )n . \label{i1}
\end{align}
\begin{align}
I_2 & \le \sum_{\hat{a}_i \in \mathcal{A}_{i,L>1}^\circ} \Pr\{ \hat{a}_i \}  \log | 2^{n} |\nonumber\\
& = n  \sum_{\hat{a}_i \in\mathcal{A}^\circ_{i,L>1}  } \Pr \{\hat{a}_i\} \nonumber \\
& \le n \frac{|\mathcal{A}_{i,L>1}| 2^{(k-1)n} }{2^{kn} (1-\epsilon') } \nonumber \\
& \stackrel{(b)}{\le} \frac{\epsilon'}{(1-\epsilon')(1-l\epsilon')}n,
 \label{i2}
\end{align}
where (b) follows from (\ref{L>1}). Finally, by definition if $\hat{a}_i \in \mathcal{A}^\circ_{i,L=1}$, then there is a unique $\hat{b}_i \in \mathcal{B}_i^\circ$ such that for all messages $m \in \mathcal{M}(\hat{a})$, it follows that $b_i(m) =\hat{b}_i$. Therefore $I_3 = 0$. Substituting (\ref{i2}) and (\ref{i1}) into (\ref{i1i2i3}) we have
\begin{align*}
H(b_i|a_i) < \left(\frac{1}{l} +l\epsilon' \right)  n +  \frac{\epsilon'}{(1-\epsilon')(1-l\epsilon')}n.
\end{align*}
Together with (\ref{hb}), it follows
\begin{align}
I(a_i;b_i) &> \left(1 - \epsilon' - \frac{1}{l} \right) (1-l\epsilon')(n + \log(1-\epsilon')) \nonumber
\\ & \hspace{13mm} - \left(\frac{1}{l} +l\epsilon' \right)  n - \frac{\epsilon'}{(1-\epsilon')(1-l\epsilon')}n. \label{iab}
\end{align}

\vspace{1mm}

\noindent\textbf{Step 3: Connect $z'_i$ and $b_i$.}

Next we show that $I(b_i;z'_i)/n \to 1$, as $\epsilon \to 0$, $n \to \infty$, for $i=1,...,k$, by upper bounding $H(z'_i|b_i)$. We first make some useful observations.
\begin{lemma}\label{lfour}
Let $\mathcal{M}(\hat{z}'_i) = \{m \in \mathcal{M}^\circ: z'_i(m) = \hat{z}'_i\}$, then for any  $m_1, m_2 \in \mathcal{M}(\hat{z}'_i)$ such that $b_i(m_1) \ne b_i(m_2)$, there exists an element of $\mathcal{B}^\times$ that will be decoded by the terminal to either $m_1$ or $m_2$.
\end{lemma}
\begin{proof}
Consider any $m_1, m_2 \in \mathcal{M}(\hat{z}'_i)$ such that $b_i(m_1) \ne b_i(m_2)$. Let $\bm{r}_1$ be the error pattern that changes the signal on $x_i$ to be $x_i(m_2)$, and let $\bm{r}_2$ be the error pattern that changes the signal on $y_i$ to be $y_i(m_1)$. Then if $m_1$ is transmitted by the source and $\bm{r}_1$ happens, the node $B_i$ will receive the same inputs as in the situation that $m_2$ is transmitted and $\bm{r}_2$ happens. Therefore $b_i(m_1, \bm{r}_1) = b_i(m_2, \bm{r}_2)$, and so either $b_i(m_1, \bm{r}_1)  \ne b_i(m_1)$ or $b_i(m_2, \bm{r}_2) \ne b_i(m_2)$ because by hypothesis $b_i(m_1) \ne b_i(m_2)$. Consider the first case that $b_i(m_1, \bm{r}_1)  \ne b_i(m_1)$, then the tuple of signals $(b_1(m_1, \bm{r}_1), ..., b_k(m_1, \bm{r}_1)) = (b_1(m_1), ..., b_i(m_1, \bm{r}_1), ..., b_k(m_1))$ will be decoded by the terminal to message $m_1$ because by hypothesis $m_1 \in \mathcal{M}^\circ$ which is decodable under any error pattern $\bm{r} \in \mathcal{R}_\mathcal{A}$. Therefore it is an element of $\mathcal{B}^\times$ since it does not equals $\bm{b}(m_1)$.
Similarly in the latter case, $b_1(m_2, \bm{r}_2) \ne b_1(m_2)$, then $(b_2(m_2, \bm{r}_2), ..., b_k(m_2, \bm{r}_2)) = (b_2(m_2), ..., b_i(m_2, \bm{r}_2), ..., b_k(m_2))$ is an element of $\mathcal{B}^\times$ and will be decoded by the terminal to $m_2$.
Therefore in either case we find an element of $\mathcal{B}^\times$ that will be decoded to either $m_1$ or $m_2$.
\end{proof}

\begin{lemma}\label{lfive}
Let $\mathcal{M}(\hat{z}'_i, \hat{b}_i) = \{ m \in \mathcal{M}^\circ: z'_i(m) = \hat{z}'_i, b_i(m) = \hat{b}_i  \}$, $\hat{b}_{i, \hat{z}'_i} = \arg \max_{\hat{b}_i \in \mathcal{B}^\circ_i } | \mathcal{M}(\hat{z}'_i, \hat{b}_i) |$ and $N(\hat{z}'_i) = |\mathcal{M}(\hat{z}_i')|$, then there are at least $\frac{1}{2} (N(\hat{z}'_i) - |\mathcal{M}(\hat{z}'_i,\hat{b}_{i,\hat{z}'_i})|)$ distinct elements of $\mathcal{B}^\times$ that will be decoded by the terminal to some messages in $\mathcal{M}(\hat{z}'_i)$.
\end{lemma}
\begin{proof}
Let $\mathcal{W} := \mathcal{M}(\hat{z}'_i)$, and
we describe an iterative procedure as follow. Pick arbitrary $m_1, m_2 \in \mathcal{W}$ such that $b_i(m_1) \ne b_i(m_2)$, and then delete them from $\mathcal{W}$ and repeat until there does not exist such $m_1, m_2$. By Lemma \ref{lfour}, each pair of elements deleted from $\mathcal{W}$ will generate a distinct element of $\mathcal{B}^\times$, which will be decoded by the terminal to either $m_1$ or $m_2$.  After the iterative procedure terminates, it follows that $|\mathcal{W}| \le  |\mathcal{M}(\hat{z}'_i,\hat{b}_{i,\hat{z}'_i})|$, because otherwise there must exist $m_1, m_2 \in \mathcal{W}$ such that $b_i(m_1) \ne b_i(m_2)$. Therefore at least $N(\hat{z}'_i) - |\mathcal{M}(\hat{z}'_i,\hat{b}_{i,\hat{z}'_i})|$  elements are deleted and the lemma is proved.
\end{proof}

We are now ready to upper bound $H(z'_i|b_i)$. Recall that $\mathcal{Z}'^\circ_i = \{ z'_i(m) : m \in \mathcal{M}^\circ   \}$. We have,

\begin{align}
H(b_i|z'_i) & = - \sum_{\hat{z}'_i \in \mathcal{Z}'^\circ_i} \Pr\{ \hat{z}'_i \} \sum_{\hat{b}_i \in \mathcal{B}^\circ_{i}} \Pr \{ \hat{b}_i | \hat{z}'_i \} \log \Pr\{ \hat{b}_i | \hat{z}'_i \}\nonumber\\
& = I_4  + I_5, \label{i4i5}
\end{align}
where
\begin{align}
I_4 & = - \sum_{\hat{z}'_i \in \mathcal{Z}'^\circ_i} \Pr\{ \hat{z}'_i \}  \Pr \{ \hat{b}_{i,\hat{z}'_i} | \hat{z}'_i \} \log  \Pr \{ \hat{b}_{i,\hat{z}'_i} | \hat{z}'_i \} \nonumber \\
& < \sum_{\hat{z}'_i \in \mathcal{Z}'^\circ_i} \Pr\{ \hat{z}'_i \} \le 1
\label{i4}
\end{align}
and
\begin{align}
I_5  & =  - \sum_{\hat{z}'_i \in \mathcal{Z}'^\circ_i} \Pr\{ \hat{z}'_i \} \sum_{\hat{b}_i \ne \hat{b}_{i,\hat{z}'_i}} \Pr \{ \hat{b}_i | \hat{z}'_i \} \log \Pr \{ \hat{b}_i | \hat{z}'_i \} \nonumber \\
& \le \sum_{\hat{z}'_i \in \mathcal{Z}'^\circ_i} \Pr\{ \hat{z}'_i \} \sum_{\hat{b}_i \ne \hat{b}_{i,\hat{z}'_i}} \Pr \{ \hat{b}_i | \hat{z}'_i \} \log | 2^{kn} |\nonumber\\
& = k n \sum_{\hat{z}'_i \in \mathcal{Z}'^\circ_i} \sum_{\hat{b}_i \ne \hat{b}_{i,\hat{z}'_i}} \Pr \{ \hat{z}'_i , \hat{b}_i \} \nonumber \\
& = kn \sum_{\hat{z}'_i \in \mathcal{Z}'^\circ_i} \frac{N(\hat{z}'_i) - |\mathcal{M}(\hat{z}'_i,\hat{b}_{i,\hat{z}'_i})|}{|\mathcal{M}^\circ|} \nonumber  \\
& \le \frac{kn}{2^{kn} (1-\epsilon')} \sum_{\hat{z}'_i \in \mathcal{Z}'^\circ_i} (N(\hat{z}'_i) - |\mathcal{M}(\hat{z}'_i,\hat{b}_{i,\hat{z}'_i})|) \nonumber \\
& \stackrel{(c)}{\le} \frac{kn \cdot 2 |\mathcal{B}^\times|}{2^{kn} (1-\epsilon')} \nonumber \\
& =  \frac{kn \cdot 2 \cdot \epsilon' \cdot 2^{kn}}{2^{kn} (1-\epsilon')} = \frac{2k\epsilon'}{1-\epsilon'}n, \label{i5}
\end{align}
where (c) follows from Lemma \ref{lfive}.
Substituting (\ref{i4}) and (\ref{i5}) to (\ref{i4i5}) we have
\begin{align}\label{hbz'}
H(b_i|z'_i) \le 1 +  \frac{2k\epsilon'}{1-\epsilon'}n.
\end{align}

\noindent \textbf{Step 4: Connect $a_i$ and $z_i$.}

Finally we discuss the connection between $a_i$ and $z_i$. We may assume that $z_i = a_i$ without loss of generality in the following sense. For every network code that achieves rate $k$ with error probability  $\epsilon$ in $\mathcal{I}_c$, we can modify it slightly to obtain a new code such that $z_i = a_i$, and such that the code on all other edges and at the terminal are the same as the original code. This modification is feasible because the encoding function $z_i$ of the original code is a function of $a_i$, and so if we let $z_i = a_i$, then the node $s_i$, $i=1,..,k$, always has enough information to reproduce the original network code. Since from the perspective of the terminal, the modified code
is the same as the original code, it also achieves rate $k$ with error probability $\epsilon$ in $\mathcal{I}_c$.
Hence,
\begin{align}\label{iaz}
I(a_i;z_i) = H(a_i). 
\end{align}

\noindent \textbf{Step 5: Connect $z_i$ and $z'_i$.}

By (\ref{iab}), (\ref{hbz'}), (\ref{iaz}) and Lemma \ref{info}, we have
\begin{align}
I(z_i;z'_i) & \ge I(a_i;z_i) + I(a_i;b_i) + I(b_i;z'_i) - H(a_i) - H(b_i)\nonumber\\
& = I(a_i;b_i) + I(b_i;z'_i) - H(b_i) \nonumber\\
& = I(a_i;b_i) - H(b_i|z'_i)\nonumber\\
& > \left(1 - \epsilon' - \frac{1}{l} \right) (1-l\epsilon')(n + \log(1-\epsilon'))   \nonumber\\
&  \hspace{15mm} - \left(\frac{1}{l} +l\epsilon' \right)  n - \frac{\epsilon'}{(1-\epsilon')(1-l\epsilon')}n  \nonumber\\
 & \hspace{40mm} - 1 -  \frac{2k\epsilon'}{1-\epsilon'}n.  \label{izz'}
\end{align}

Recall that (\ref{izz'}) is obtained under the assumption that the source message is uniformly distributed over $\mathcal{M}^\circ$, and hence $\bm{a}$ is uniformly distributed over $\mathcal{A}^\circ$. Therefore the random variables $\{a_1(M^\circ),...,a_k(M^\circ)\}$ are not independent. Now we consider the case that $\bm{a}$ is uniformly distributed over $[2^n]^k$ and the network code is the same as before. Specifically, the decoding function and the encoding functions at all edges except $a_1, ..., a_k$ are the same as the network code that achieves rate $k$ with error probability $\epsilon$ in $\mathcal{I}_c$. We are interested in the mutual information between $z_i$ and $z'_i$ under this setting where  the random variables $\{a_1,...,a_k\}$ are independent. Let $\bm{a} \gets [2^n]^k$ denote that the distribution of $\bm{a}$ is uniform over $[2^n]^k$,
\begin{align}
I_{\bm{a} \gets [2^n]^k} (z_i;z'_i) = \Pr_{\bm{a} \gets [2^n]^k}\{\bm{a} \in \mathcal{A}^\circ\} I_{\bm{a} \gets [2^n]^k} (z_i;z'_i| \bm{a} \in \mathcal{A}^\circ) & \nonumber  \\
& \hspace{-73mm} + \Pr_{\bm{a} \gets [2^n]^k} \{ \bm{a} \in \mathcal{A}^\times \} I_{\bm{a} \gets[2^n]^k } (z_i;z'_i| \bm{a} \in \mathcal{A}^\times ).
\end{align}
Note that
\begin{align}\label{ft}
I_{\bm{a} \gets [2^n]^k} (z_i;z'_i| \bm{a} \in \mathcal{A}^\circ) =  I_{\bm{a} \gets \mathcal{A}^\circ}(z_i;z'_i)= I(z_i;z'_i),
\end{align}
which is exactly the result we computed in (\ref{izz'}). And
\begin{align}\label{st}
\Pr_{\bm{a} \gets [2^n]^k} \{\bm{a} \in \mathcal{A}^\circ\} = |\mathcal{A}^\circ|/2^{kn} \ge (1-\epsilon').
\end{align}
Therefore by (\ref{ft}) and (\ref{st}), $I_{\bm{a} \gets [2^n]^k} (z_i;z'_i) \ge (1-\epsilon') I(z_i;z'_i)$. Then by (\ref{izz'}), for any $\epsilon_1 > 0$,  by first choosing a sufficiently large $l$ and then a sufficiently small $\epsilon$ and a sufficiently large $n$, it follows that $I_{\bm{a} \gets [2^n]^k} (z_i;z'_i)/n > 1-\epsilon_1$, $i=1,...,k$. Hence unit rate is asymptotically feasible in $\mathcal{I}$ by the channel coding theorem \cite{Cover:2006gz}. This completes the proof of the first part of the theorem.

``$\Leftarrow$''. Conversely, we show that if unit rate is feasible in $\mathcal{I}$, then rate $k$ is feasible in $\mathcal{I}_c$. Again we use the constructive scheme in Figure \ref{ach}. In $\mathcal{I}_c$, the source lets $M = (M_1,...,M_k)$, where the $M_i$'s are i.i.d. uniformly distributed over $[2^{n}]$. Let the network code be $a_i(M)=x_i(M)=y_i(M)=z_i(M)=z'_i(M)=M_i$, $i=1,...,k$, and let node $B_i$, $i=1,...,k$ performs majority decoding. The terminal $t$ will not decode an error as long as the multiple-unicast instance $\mathcal{I}$ does not commit an error. This happens with probability at least $1-\epsilon$, which implies the feasibility of rate $k$ in $\mathcal{I}_c$.

\section{conclusion}
\begin{figure*}[htb!]
\center
  \includegraphics[width=0.75\textwidth]{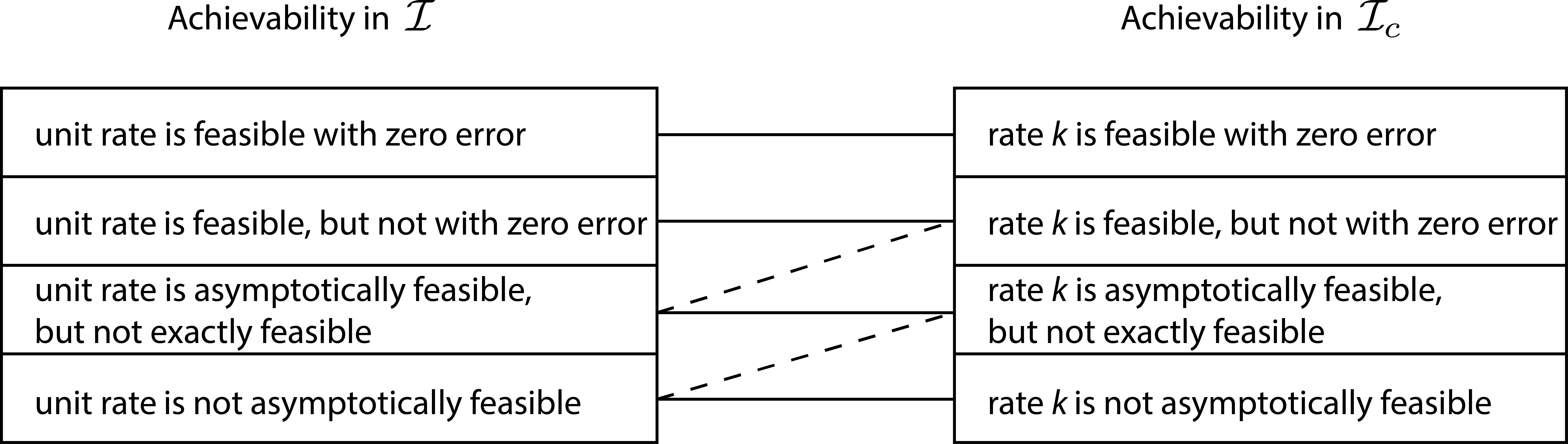}
  \caption{Equivalence between multiple-unicast and network error correction under the construction of Fig. \ref{zeroerr}.  In our reduction, given an instance $\mathcal{I}$ of the multiple-unicast network coding problem, we construct an instance $\mathcal{I}_c$ of the single-unicast network error correction problem.  This figure expresses our current understanding of the relation between achievability in $\mathcal{I}$ and $\mathcal{I}_c$. A solid line between two states means that there exist instances of $\mathcal{I}$ and $\mathcal{I}_c$ with the corresponding achievability. If there is no line between two states, there do not exist instances of $\mathcal{I}$ and $\mathcal{I}_c$ with the corresponding achievability. A dashed line between two states means that whether there is a solid line between the two states or not is still  an open problem. }\label{result}
\end{figure*}
We summarize the results of this paper in Fig. \ref{result} which expresses the possible connections between the feasibility of a general multiple-unicast network coding instance $\mathcal{I}$ and its reduced single-unicast network error correction instance $\mathcal{I}_c$. Our results present an equivalence between multiple-unicast and network error correction under zero-error communication.
Namely, determining the feasibility of zero-error unit rate in $\mathcal{I}$ is equivalent to determining  the feasibility of zero-error rate $k$ in $\mathcal{I}_c$.
This is expressed in Fig. \ref{result} by the fact that the two states in the first row are connected by a single solid edge and there are no other edges connected to these two states.

For the vanishing error model, however, the implication of our results are more involved.
In our reduction there is a slight slackness, which gives rise to the two dashed lines in Fig. \ref{result}. 
For example, consider the case that unit rate is feasible (but not with zero error) in $\mathcal{I}$.
Then it follows that  rate $k$ is feasible (but not with zero error) in $\mathcal{I}_c$. This fact is expressed by the sole solid line leaving the state that unit rate is feasible (but not with zero error) in $\mathcal{I}$.
However, if rate $k$ is feasible (but not with zero error) in
$\mathcal{I}_c$, then there are potentially two possibilities: (a) that unit
rate is feasible (but not with zero error) in $\mathcal{I}$; and (b) that
unit rate is asymptotically feasible  (but not exactly feasible) in
$\mathcal{I}$. Option (a) is represented by a solid line, as indeed we have
observed instance $\mathcal{I}$ with a corresponding $\mathcal{I}_c$ that
fits this setting. Option (b) is represented by a dashed line, as on one the
one hand it has not been ruled out by our analysis, but on the other hand we are not aware of how to construct instances $\mathcal{I}$ with a corresponding $\mathcal{I}_c$ that fit this setting. 

All in all, under the vanishing error model, the two dashed lines in Fig. \ref{result} do not allow us to directly determine the feasibility of  $\mathcal{I}$ based on the feasibility of $\mathcal{I}_c$. Nevertheless, we may consider the following problem on $\mathcal{I}$ which can be solved by the study of $\mathcal{I}_c$ using our results: Given an instance $\mathcal{I}$, partially determine between the three possible settings in the following manner:
if unit rate is feasible (but not with zero error) in $\mathcal{I}$, answer {\em yes}; if unit rate is not asymptotically feasible in $\mathcal{I}$, answer {\em no}; if unit rate is asymptotically feasible but not exactly feasible in $\mathcal{I}$, then any answer is considered correct. By our results,  answering {\em yes} if and only if rate $k$ is feasible (but not with zero error) in $\mathcal{I}_c$ solves the problem above.
Whether the partial distinction problem above on $\mathcal{I}$ is a {\em difficult} one (compared to the standard feasibility of multiple-unicast network coding) is yet to be established.

Finally we note that our reduction and analysis in the paper consider the
case that the connections between source destination pairs in $\mathcal{I}$ have unit rate.
Our reduction can be adapted to connections with different rates.

\section*{Acknowledgement}
This work has been supported in part by NSF grant CCF-1440014, CCF-1440001, 	CCF-1439465, and CCF-1321129.

\bibliographystyle{IEEEtran}
\bibliography{ref}
\end{document}